\documentclass[letter,11pt]{article}
\usepackage[top=1in, bottom=1in, left=1in, right=.95in]{geometry}
\usepackage{amsmath,amsthm}
\usepackage{amssymb}
\usepackage{authblk}
\usepackage{paralist}
\usepackage[ruled]{algorithm}
\usepackage[noend]{algpseudocode}
\usepackage{hyperref}

\usepackage{epsfig}
\usepackage{enumitem}

\usepackage{hyperref}

\usepackage[normalem]{ulem}

\newcommand{\comment}[1]{}

\usepackage{etoolbox}\AtBeginEnvironment{algorithmic}{\small}

\comment{
	\newtheoremstyle{redstyle}
	{3pt}
	{3pt}
	{\color{black}}
	{}
	{\color{red}\bfseries}
	{:}
	{.5em}
	{}
	\theoremstyle{redstyle} 
}

\usepackage{color}

\newif\ifarticle
\articletrue

\ifarticle

\newtheorem{theorem}{Theorem}
\newtheorem{corollary}{Corollary}
\newtheorem{lemma}{Lemma}
\fi
\newtheorem{claim}{Claim}
\newtheorem{fact}{Fact}

\newtheorem{proposition}{Proposition}

\newif\ifrobocza
\roboczafalse

\newif\ifarxiv
\arxivtrue

\ifrobocza
\newcommand{\labell}[1]{\label{#1}\marginpar{#1}}
\newcommand{\tj}[1]{{\color{blue}{#1}}}
\newcommand{\kn}[1]{{\color{green}{#1}}}
\newcommand{\xx}[1]{{\color{red}{#1}}}
\else
\newcommand{\labell}[1]{\label{#1}}
\newcommand{\tj}[1]{#1}
\newcommand{\kn}[1]{#1}
\newcommand{\xx}[1]{#1}
\fi

\newcommand{\prob}{\text{Prob}}
\newcommand{\NAT}{{\mathbb N}}

\renewcommand{\thefootnote}{\fnsymbol{footnote}}

\newcommand{\remove}[1]{}

\newcommand{\CC}{\text{CC}}

\newcommand{\MST}{\text{MST}}

\newif\iffull
\fullfalse

\begin{document}

	\title{MST in $O(1)$ Rounds of Congested Clique
	}

	\author{Tomasz Jurdzi\'{n}ski\footnote{email: \url{tju@cs.uni.wroc.pl}}
	}
	\author{Krzysztof Nowicki\footnote{email: \url{knowicki@cs.uni.wroc.pl}} 
	}
	
	\affil{Institute of Computer Science, University of Wroc{\l}aw, Poland}

	\date{}
	
	\maketitle 
	
	\begin{abstract}

We present a distributed randomized algorithm finding \textit{Minimum Spanning Tree} (MST) of a given graph in $O(1)$ rounds, with high probability, in the congested clique model.

The input graph in the congested clique model is a graph of $n$ nodes, where each node initially knows only its incident edges. The communication graph is a clique with limited edge bandwidth: each two nodes (not necessarily neighbours in the input graph) can exchange $O(\log n)$ bits.

As in previous works, the key part of the \MST{} algorithm is an efficient \textit{Connected Components} (CC) algorithm. However, unlike the former approaches, we do not aim at simulating the standard Boruvka\tj{'s} algorithm, at least at initial stages of the \CC{} algorithm. Instead, we develop a new technique which combines connected components of sample sparse subgraphs of the input graph in order to accelerate the process of uncovering connected components of the original input graph. More specifically, we develop a sparsification technique which reduces an initial \CC{} problem in $O(1)$ rounds to its two restricted instances. The former instance has a graph with maximal degree $O(\log \log n)$ as the input -- here our sample-combining technique helps. In the latter instance, a partition of the input graph into $O(n/\log\log n)$ connected components is known. This gives an opportunity to apply previous algorithms to determine connected components in $O(1)$ rounds.

\tj{Our result addresses a problem proposed by Lotker et al.\ [SPAA 2003; SICOMP 2005] and} improves over previous $O(\log^* n)$ algorithm of Ghaffari et al.\ [PODC 2016], and $O(\log \log \log n)$ algorithm of Hegeman et al.\ [PODC 2015]. It also determines $\Theta(1)$ round complexity in the congested clique for \MST{}, as well as other graph problems, including bipartiteness, cut verification, s-t connectivity, and cycle containment.
\end{abstract}

	
	\small
		\noindent

	
	
\renewcommand{\thefootnote}{\arabic{footnote}}

\newpage

\section{Introduction}\labell{s:intro}
The congested clique model of distributed \tj{computing} attracted much attention in algorithmic community in recent years. \tj{Initially, each node knows its incident edges in the input graph $G(V,E)$. Unlike the classical CONGEST model \cite{peleg-book}, the communication graph connects each pair of nodes, even if they are not neighbours in the input graph, i.e., in each round, each pair of nodes can exchange a message of size $O(\log n)$ bits.}

The main algorithm-theoretic motivation of the model is to understand the role of congestion in distributed computing. The well-known LOCAL model of distributed computing ignores congestion by allowing unlimited size of transmitted messages \cite{peleg-book,Linial92} and focuses on \emph{locality}. The CONGEST model on the other hand takes congestion into account by limiting the size of transmitted messages. Simultaneously, locality plays an important role as well in the CONGEST model, since direct communication is possible only between neighbors of the input graph. The congested clique is considered as a complement, which focuses \tj{solely} on congestion.

Some variants and complexity measures for the congested clique have applications to efficiency of algorithms in other models adjusted to current computing challenges such as $k$-machine big data model \cite{KlauckNP015} or MapReduce \cite{HegemanP15,KarloffSV10}.

\subsection{Related work}
The general congested clique model as well as its limited variant called broadcast congested clique were studied in several papers, e.g. \cite{Lotker:2003:MCO:777412.777428,HegemanPPSS15,GhaffariParter2016,DruckerKO13,BeckerMRT14,Lenzen:2013:ODR:2484239.2501983,MT2016}. The recent Lenzen's \cite{Lenzen:2013:ODR:2484239.2501983} constant time routing and sorting algorithm in the unicast congested clique exhibited strength of this model and triggered a new wave \kn{of research}.

Besides general interest in the congested clique, specific attention has been given to \MST{} and connectivity. Lotker et al.\ \cite{Lotker:2003:MCO:777412.777428} \kn{proposed} a $O(\log \log n)$ round deterministic algorithm for \MST{} in the unicast model. An alternative solution of the same complexity has been presented recently \cite{Korhonen16}. The best known randomized solution for \MST{} works in $O(\log^*n)$ rounds \cite{GhaffariParter2016}, improving the recent $O(\log\log\log n)$ bound \cite{HegemanPPSS15}. The result from \cite{HegemanPPSS15} uses a concept of linear graph sketches \cite{AhnGM12}, while the authors of \cite{GhaffariParter2016} introduce new sketches, which are sensitive to the degrees of nodes and adjusted to the congested clique model. Reduction of the number of transmitted messages in the \MST{} algorithms was studied in \cite{DBLP:conf/fsttcs/PS16}. In contrast to the general (unicast) congested clique, no sub-logarithmic round algorithm is known for \MST{} in the broadcast congested clique, while the first sub-logarithmic solution for the \CC{} problem has been obtained only recently \cite{JurdzinskiN17}. An extreme scenario that the algorithm consists of one \kn{round} in which each node can send only one message has also been considered. As shown in \cite{AhnGM12}, connectivity can be solved with public random bits in this model, provided nodes can send messages of size $\Theta(\log^3 n)$.

In \cite{DruckerKO13}, a simulation of powerful classes of bounded-depth circuits in the congested clique is presented, which points out to the power of the congested clique and explains difficulty in obtaining lower bounds.

\subsection{Our result}
The main result of the paper determines $O(1)$ round complexity of the \MST{} problem in the congested clique.
\begin{theorem}
\labell{t:MST} 
There is a randomized algorithm in the congested clique model that computes a minimum spanning tree in $O(1)$ rounds, with high probability.
\end{theorem}
Using standard reductions of some graph problems to the connectivity problem, we establish $O(1)$ round complexity of several graph problems in the congested clique model.
\begin{corollary}
\labell{c:reductions}
There are randomized distributed algorithms that solve the following verification problems in the congested clique model in $O(1)$ rounds with high probability: bipartiteness verification, cut verification, s-t connectivity, and cycle containment. 
\end{corollary}

\subsection{Preliminaries}\labell{s:prel}
In this section we provide some terminology and tools for design of distributed algorithms in the congested clique.

Given a natural number $p$, $[p]$ denotes the set $\{1,2,\ldots,p\}$.

We use the following \emph{Lenzen's routing} result in our algorithms.
\begin{lemma}\cite{Lenzen:2013:ODR:2484239.2501983}
\labell{l:lenzen}
Assume that each node of the congested clique is given a set of $O(n)$ messages with fixed destination nodes. Moreover, each node is the destination of $O(n)$ messages from other nodes. Then, it is possible to deliver all messages in $O(1)$ rounds of the congested clique.
\end{lemma}
Efficient congested clique algorithms often make use of auxiliary nodes. Below, we formulate this opportunity to facilitate design of algorithms.
\begin{lemma}\labell{l:auxiliary}
Let $A$ be a congested clique algorithm which, except of the nodes $u_1,\ldots,u_n$ corresponding to the input graph, uses $O(n)$ auxiliary nodes $v_1,v_2,\ldots$ such that the auxiliary nodes do not have initially any knowledge of the input graph on the nodes $u_1,\ldots,u_n$. Then, each round of $A$ might be simulated in $O(1)$ rounds in the standard congested clique model, without auxiliary nodes.
\end{lemma}
\begin{proof}
Assume that there are at most $cn$ auxiliary nodes, for constant $c\in\NAT$. We can assign the set of $c$ auxiliary nodes $V_j=\{v_{(j-1)c+1},\ldots,v_{jc}\}$ to $u_j$ for each $j\in[n]$ and assign internal IDs in the range $[c]$ to the elements of $V_j$. Then, a round of an original algorithm with auxiliary nodes is simulated in $c^2$ actual rounds indexed by $(a,b)\in [c]^2$. The $a$th auxiliary nodes assigned to $u_j$ transmit messages addressed to the $b$th auxiliary nodes \kn{of $u_1,\ldots,u_n$} in the rounds indexed $(a,b)$.
\end{proof}

We consider randomized algorithms in which a computational unit in each node of the input network can use private random bits in its computation. We say that some event holds with high probability (whp) for an algorithm $A$ running on an input of size $n$ if this event holds with probability $1-1/n^c$ for a given constant $c$. We require here that the constant $c$ can be chosen arbitrarily large, without changing \emph{asymptotic} complexity of the considered algorithm.

\subsubsection{Graph terminology}
If not stated otherwise, we consider undirected graphs. Thus, in particular, an edge $(u,v)$ appears in the graph iff $(v,u)$ appears in that graph as well. For a node $v\in V$ of a graph $G(V,E)$, $N(v)$ denotes the set of neighbours of $v$ in $G$ and $d(v)$ denotes the degree of $v$, i.e., $d(v)=|N(v)|$. We say that a graph $G(V,E)$ has degree $\Delta$ if the degree of each node of $G$ is smaller or equal to $\Delta$.

A \emph{component} of a graph $G(V,E)$ is a connected subgraph of $G$. That is, $C\subseteq V$ corresponds to a \emph{component} of $G$ iff the graph $G(C, E\cap (C\times C))$ is connected. A component $C$ of a graph $G(V,E)$ is \emph{growable} if there are edges connecting $C$ with $V\setminus C$, i.e., the set $E\cap (C\times(V\setminus C))$ is non-empty. Otherwise, the component $C$ is \emph{ungrowable}.

For a graph $G(V,E)$, sets $C_1, C_2, ..., C_k \subset V$ form a partition of $G$ into \emph{components} if $C_i$s are pairwise disjoint, $\bigcup_{i\in[k]} C_i = V$, and $C_i$ is a component of $G(V,E)$ for each $i\in[k]$. A partition $C_1,\ldots,C_k$ of a graph $G(V,E)$ into components is the \emph{complete partition} if $C_i$ is ungrowable for each $i\in[k]$. For a fixed $E'\subseteq E$, the \emph{complete partition} of $G(V,E')$ will be also called the \emph{complete partition with respect to} $E'$.

Given a partition $\mathbb{C}$ of a graph $G(V,E)$ into components and $v\in V$, $C(v)$ denotes the component of $\mathbb{C}$ containing $v$.

An edge $(u,v)$ is \emph{incident} to a component $C$ wrt to a partition $\mathbb{C}$ of a graph if it connects $C$ with another component of $\mathbb{C}$, i.e., $C(u)\neq C(v)=C$ or $C(v)\neq C(u)=C$.

\subsubsection{Graph problems in the congested clique model}

Graph problems in the congested clique model are considered in the following framework. The joint input to the $n$ nodes of the network is an undirected $n$-node weighted graph $G(V,E,c)$, where each node in $V=\{u_1,\ldots,u_n\}$ corresponds to a node of the communication network and weights $c(e)$ of edges are integers of polynomial size (i.e., each weight is a bit sequence of length $O(\log n)$). Each node $u_i$ initially knows the network size $n$, its unique ID $i\in[n]$, the list of IDs of its neighbors in the input graph and the weights of its incident edges. Specifically, $\text{ID}(v)=i$ for $v=u_i$. All graph problems are considered in this paper in accordance with this definition.

\subsubsection{Connected Components and Minimum Spanning Tree}
In the paper, we consider the connected components problem (\CC{}) and the minimum spanning tree problem (MST). A solution for the \CC{} problem consists of the complete partition of the input graph $G(V,E)$ into connected components $C_1\cup\cdots\cup C_k=V$, accompanied by spanning trees of all components. Our goal is to compute \CC{} or \MST{} of the input graph, i.e., each node should know the set of edges inducing \CC{}/\MST{} at the end of an execution of an algorithm.

For the purpose of fast simultaneous executions of many instances of the \CC{} algorithms, we also consider the definition of the \CC{} problem, where spanning trees of all components are known only to a fixed node of a network. The presented solutions usually correspond to this weaker definition. However, for a single instance of the \CC{}/\MST{} problem, a spanning forest can be made known to all nodes in two rounds, provided it is known to a fixed node $v$. Namely, it is sufficient that $v$ fixes roots of spanning trees of all components. Then, in the former round, $v$ sends to each $u$ the ID of the parent of $u$ in the appropriate tree. In the latter round each $u$ sends the ID of its parent to all nodes of the network.

We also consider the situation that some ``initial'' partition $\mathbb{C}$ of the input graph into connected components is known at the beginning of an execution of an algorithm to a fixed node.

We say that a component $C$ of a partition $\mathbb{C}$ which is known to all/some nodes of the network is an \emph{active} component if $C$ is a growable component of the original input graph $G(V,E)$. Otherwise, if $C$ is not a growable component of $G(V,E)$, $C$ is an \emph{inactive} component.

Below, we make a simple observation which we use in our algorithms.

\begin{fact}\labell{f:graph:union}
Assume that solutions of the \CC{} problem for the graphs $G(V,E_1)$ and $G(V,E_2)$ are \xx{available.} Then, one can determine a solution of the \CC{} problem for $G(V,E_1\cup E_2)$. 
\end{fact}
\section{High-level description of our solution}
In this section we describe our \MST{} algorithm on a top level. The main technical result is an $O(1)$ round connectivity algorithm. The extension to \MST{} \tj{(described in Section~\ref{s:MST})} is based on a known technique \kn{which requires} $n^{1/2}$ simultaneous executions of the connectivity algorithm \xx{(on partially related instances)}. Thus, the key issue in design of the $\MST{}$ algorithm is to guarantee that such $n^{1/2}$ simultaneous executions of the connected components algorithm can be performed in $O(1)$ rounds. 

The algorithm for connected components works in two phases: Sparsification Phase and Size-reduction Phase. In Sparsification Phase, we reduce the original connectivity problem to two specific instances of the \CC{} problem (Lemma~\ref{l:degreereduction}). \xx{(The idea of this reduction can be implemented in the much weaker broadcast congested clique model and allows to obtain new round-efficient algorithms in that model \cite{JurdzinskiN17}.)}

In the former instance, a partition of the input graph  into $O(n/\log\log n)$ active and some non-active components is known. The connected components can be determined for such an instance in $O(1)$ rounds by the algorithm from \cite{GhaffariParter2016}. The latter instance is a graph with degree $O(\log\log n)$. Therefore, Size-reduction Phase gets a graph with degree $O(\log\log n)$ as the input. In this phase, the \CC{} problem for such a sparse input graph is reduced to an instance of the \CC{} problem where, except of the input graph $G$, a partition of $G$ into $O(n/\log\log n)$ active components is known. Therefore, as before, the connected components can be determined for this final instance in $O(1)$ rounds by the algorithm from \cite{GhaffariParter2016}.

\paragraph{Connected Components: Sparsification Phase.}
Sparsification Phase is based on a simple \emph{deterministic} procedure (see Alg.~\ref{a:degreereduction}):
\begin{itemize}
\item Firstly, for each node $v$, an edge $(u,v)$ connecting $v$ with its highest degree neighbour is determined and delivered to a fixed node called the \emph{coordinator}. 
\item Then, the complete partition $\mathbb{C}$ with respect to the set of edges delivered to the coordinator is computed. The \emph{degree} of each component of $\mathbb{C}$ is defined as the largest degree of its elements.
\item Next, for each node $v$, an edge $(u,v)$ connecting $v$ with the highest degree component $C\neq C(v)$ of $\mathbb{C}$ is determined and send to the coordinator.
\item The coordinator computes the complete partition $\mathbb{C}'$ with respect to the set of edges announced in all steps of this procedure, assigns IDs to the components of this partition. Then, the coordinator sends to each node the ID and the degree of its component. Finally, the nodes pass information obtained from the coordinator to their neighbours.
\end{itemize}
Let $C$ be a component of $\mathbb{C}'$. We say that $C$ is \emph{awake} if the degree of $C$ (i.e., the largest degree of its elements) is at most $s$, for some fixed $s\in\NAT$. Otherwise, $C$ is \emph{asleep}. A node $u$ is awake (asleep, respectively) iff $u$ belongs to an awake (asleep, respectively) component. \kn{As we will show, each awake component of $\mathbb{C}'$ has size $\ge s$ and therefore the partition $\mathbb{C}'$ contains $O(n/s)$ awake components.} Moreover, the degree of the graph induced by edges incident to nodes from asleep components is smaller than $s$. (This fact does not follow simply from the definitions, since the graph contains also neighbors of nodes from asleep components located in awake components).

Using the above properties for $s=\log\log n$ we can split the input graph $G(V,E)$ into a subgraph $G_A$ containing $O(n/\log\log n)$ growable components and a subgraph $G_B$ with degree $O(\log\log n)$.  For the former subgraph, we determine connected components in $O(1)$ rounds using the algorithm from \cite{GhaffariParter2016}, based on graph sketches.  The connected components of the latter subgraph are determined in Size-reduction Phase.

\paragraph{Connected Components: Size-reduction Phase.} 
The key technical novelty in our solution is an algorithm which reduces the \CC{} problem for a graph of degree $O(\log\log n)$ to the \CC{} problem for a graph with $O(n/\log\log n)$ components. This algorithm is the main ingredient of Size-reduction Phase.
A pseudocode of the algorithm is presented in Alg.~\ref{a:componentreduction}. We find connected components of the described above graph 
$G_A$, by applying this reduction (Alg.~\ref{a:componentreduction}) and the algorithm from \cite{GhaffariParter2016} (which finds connected components for graphs with a known partition into $O(n/\log\log n)$ components, in $O(1)$ rounds).

In order to describe the above reduction, assume that the (upper bound on) degree of an input graph $G(V,E)$ is $\Delta=O(\log\log n)$. The idea of the reduction is to calculate simultaneously $m=n^{1/2}$ spanning forests for randomly chosen sparse subgraphs $G_i=G(V,E_i)$  of the input graph (called \emph{samples}), and use the results to build a partition of $G$ into $O(n / \log \log n)$ growable components and some non-growable components. Below, we describe the reduction in more detail.

First, we build random subgraphs $G_i=G(V,E_i)$ of $G$ for $i\leq m=\sqrt{n}$ such that,
for each $i\in[m]$, the set of edges $E_i$ can be collected at a single node. As a node can receive only $O(n)$ messages in a round, the size of $E_i$ should be $O(n)$ as well. To assure this property, each edge will belong to $E_i$ for each $i\in[m]$ with  probability $1/ \log \log n$, and the random choices for each $i\in[m]$ and each edge are independent. For each $i\in[m]$, all edges from $E_i$ are sent to a fixed node $b_i$ (the $i$th boss) and $b_i$ computes connected components of $G_i$ locally. \kn{The problem is that only one message per round might be transmitted on each edge, while a node can choose non-constant number of incident edges as the elements of the $i$th sample $E_i$ for some $i\in[m]$. (Thus, all those edges should be delivered to $b_i$.)} This problem is solved by the fact that, whp, the random choices defining the graphs $G_i$ require to send $O(n)$ messages and receive $O(n)$ messages by each node. (Note that the expected number of edges in $E_i$ is $O(n)$.) If this is the case, all messages can be delivered with help of Lenzen's routing algorithm \cite{Lenzen:2013:ODR:2484239.2501983} (Lemma~\ref{l:lenzen}).

Another challenge is how to combine connected components of the graphs $G_i$ such that the number of growable components is reduced to $O(n/\log\log n)$. To this aim, each node is chosen to be a leader, independently of other nodes, with probability $1/\log\log n$. Thus, the number of leaders will be $\Theta(n/ \log \log)$ whp. Then, for each node $v$, if $v$ is in a connected component containing a leader in some graph $G_i$, information about its connection with some leader will be delivered to the coordinator. Thus, if each node is connected to some leader in the partition determined by the coordinator then we have $O(n/ \log \log n)$ connected components and we are done. Certainly, we cannot get such a guarantee. However, as we show in Section~\ref{ss:leaderless}, the number of nodes from growable components which are not connected to any leader will be $O(n/\log \log n)$, whp. The main part of the proof of this fact is to determine a sequence of \emph{independent} random variables whose sum gives the upper bound on the number of nodes in growable components which are not connected to a leader in the final partition.

\paragraph{Minimum Spanning Tree.}
As shown in \cite{Karger:1995:RLA:201019.201022,HegemanPPSS15}, it is possible to reduce the \MST{} problem of an input graph, to two instances of \MST{} on graphs with $O(n^{3/2})$ edges. Then, \MST{} for a graph with $O(n^{3/2})$ edges is reduced to $O(\sqrt{n})$ instances of the \CC{} problem, where the set of edges in the $i$th instance is included in the set of edges of the $(i+1)$st instance. In Section~\ref{s:MST}, we show that our algorithm can be executed in parallel on these specific $\sqrt{n}$ instances of the \CC{} problem. The main challenge here is that a ``naive'' implementation of these parallel executions requires to send superlinear number of messages by some nodes. However, using Lenzen's routing \cite{Lenzen:2013:ODR:2484239.2501983} (see Lemma~\ref{l:lenzen}) and the fact that the set of edges of the $(i+1)$st instance of \CC{} includes the set of edges of the $i$th instance for each $i\in[m]$, we show that connected components of all those instances can be computed in parallel.

\section{Connectivity in $O(1)$ rounds}
In this section we describe our \CC{} algorithm which leads to the following theorem.

\begin{theorem}
\labell{t:connectedcomponents}
There is a randomized algorithm in the congested clique model that computes connected components in $O(1)$ rounds, with high probability.
\end{theorem}

The algorithm consists of two phases: Sparsification Phase and Size-reduction Phase. 

In Sparsification Phase, we reduce the original problem to two specific instances of the \CC{} problem. In the former instance, the \CC{} problem has to be solved for a graph with degree $O(\log\log n)$. The latter instance is equipped with additional information about a partition of the considered graph in $O(n/\log\log n)$ components. The pseudocode of the appropriate algorithm is presented as Alg.~\ref{a:degreereduction}. The following lemma describes the reduction more precisely. 
\begin{lemma}
\labell{l:degreereduction}
There is a deterministic algorithm in the congested clique that reduces in $O(1)$ rounds the \CC{} problem for an arbitrary graph $G(V,E)$ to the instances of the \CC{} problem for graphs $G_A(V,E_A)$ and $G_B(V,E_B)$ such that $E_A\cup E_B=E$ and
\begin{itemize}
\item a partition $\mathbb{C}_A$ of $G_A$ into $O(n/\log\log n)$ active components is known to a fixed node;
\item each node knows which of its incident edges belong to $E_A$ and which of them belong to $E_B$;
\item the degree of $G_B$ is $O(\log \log n)$.
\end{itemize}
\end{lemma}

An important building block of our solution comes from \cite{GhaffariParter2016}, where the properties of graph sketches play the key role. It is the algorithm which determines connected components of the input graph in $O(1)$ rounds, provided an initial partition of the input graph in $O(n/\log\log n)$ growable and some ungrowable components is known at the beginning. For further applications in a solution of the \MST{} problem, we state a stronger result regarding several simultaneous executions of the algorithm.

\begin{lemma}
\labell{l:GP2016}\cite{GhaffariParter2016}
There is a randomized algorithm in the congested clique model that computes connected components of a graph in $O(1)$ rounds with high probability, provided that a partition of the input graph into $O(n / \log \log n)$  growable (and some ungrowable) components is known to a fixed node at the beginning of an execution of the algorithm. Moreover, it is possible to execute $m=\sqrt{n}$ instances of the problem simultaneously in $O(1)$ rounds. 
\end{lemma}

\begin{proof}
In \cite{GhaffariParter2016}, the authors gave $O(1)$ round procedure $\mathit{ReduceCC}(x)$, reducing the number of active components from $O(n / \log^2 x)$ to $n/x$ in $O(1)$ rounds, with high probability. Thus, by using $\mathit{ReduceCC}$ constant number of times, it is possible to reduce the number of active components from $n / \log \log n$ to $0$ in $O(1)$ rounds. And, if there are no active (i.e., growable) components in a partition, then that partition describes connected components of the input graph (i.e., it is the complete partition of the input graph).
\end{proof}
Let GPReduction denote the algorithm satisfying properties from Lemma~\ref{l:GP2016}. Using GPReduction, we can determine connected components of the graph $G_A$ (described in Lemma~\ref{l:degreereduction}) in $O(1)$ rounds. Thus, in order to build $O(1)$ rounds \CC{} algorithm, it is sufficient to solve the problem for graphs with degree $O(\log\log n)$. This problem is addressed in Size-reduction Phase. In the following lemma, we show that the \CC{} problem for a $O(\log\log n)$-degree graph can be reduced to the \CC{} problem for a graph with $O(n/\log\log n)$ growable components in $O(1)$ rounds. The pseudocode of the algorithm performing this reduction is given in Alg.~\ref{a:componentreduction}.
\begin{lemma}
\labell{l:componentreduction}
There is a randomized algorithm in the congested clique model that reduces in $O(1)$ rounds the \CC{} problem for a graph with degree bounded by $\log \log n$ to an instance of the \CC{} problem for which a partition with $O(n/\log \log n)$ active connected components is known to a fixed node, with high probability.
\end{lemma}
Then, the next application of the algorithm from \cite{GhaffariParter2016} (Lemma~\ref{l:GP2016}) gives the final partition of the input graph into connected components, as summarized in Alg.~\ref{a:cc}. 

The proofs of Lemma \ref{l:degreereduction} and Lemma \ref{l:componentreduction} are presented in Section~\ref{s:degreereduction} and Section~\ref{s:componentreduction}, respectively. Using Lemmas \ref{l:GP2016}, \ref{l:degreereduction}, and \ref{l:componentreduction}, one can show that Alg.~\ref{a:cc} determines connected components of an input graph in $O(1)$ rounds, with high probability. This in turn gives the proof of Theorem~\ref{t:connectedcomponents}.

\begin{algorithm}
\caption{ConnectedComponents \Comment{$G(V,E)$ is the input graph}}
\label{a:cc}
\begin{algorithmic}[1]
\Statex \textbf{Sparsification Phase}
\State Execute Alg.~\ref{a:degreereduction} on the input graph for $s=\log\log n$ \label{s:cc1} 
\State Let $G_A$ and $\mathbb{C}_A$ be the graph and its $O(n/\log\log n)$-size partition determined in Alg.~\ref{a:degreereduction}\Comment{Lemma~\ref{l:degreereduction}}
\State $G_B\gets $ the graph of degree $O(\log\log n)$ determined in Alg.~\ref{a:degreereduction}\Comment{Lemma~\ref{l:degreereduction}}
\State Execute Alg.~GPReduction on $G_A$, using the partition $\mathbb{C}_A$\label{s:GP2016-1} \Comment{Lemma~\ref{l:GP2016}}
\Statex \textbf{Size-reduction Phase}
\State Execute Alg.~\ref{a:componentreduction} on $G_B$ 
\State $G'\gets$ the graph obtained in Alg.~\ref{a:componentreduction}, with its partition into $O(n/\log\log n)$ active components \label{s:cc7} \Comment{Lemma~\ref{l:componentreduction}}
\State Execute Alg.~GPReduction on $G'$, using its partition determined by Alg.~\ref{a:componentreduction}\label{s:GP2016-2}
\Comment{Lemma~\ref{l:GP2016}}
\State \kn{Combine connected components computed in steps \ref{s:GP2016-1} and \ref{s:GP2016-2}\Comment{Fact~\ref{f:graph:union}}}
\end{algorithmic}
\end{algorithm}

\subsection{Graph sparsification} \labell{s:degreereduction}
In this section we describe Sparsification Phase and prove Lemma \ref{l:degreereduction}. Let the \emph{coordinator} be a fixed node of the input network. For a partition $\mathbb{C}$ of a graph $G(V,E)$ into components, we use the following notations:
\begin{itemize}
\item $d(C) = \max\limits_{v \in C} d(v)$ is the \emph{degree} of the component $C$ of $\mathbb{C}$,
\item $I(v)$ is the ID of the component $C(v)$, according to a fixed labeling of components of $\mathbb{C}$.
\end{itemize}
The general idea of the reduction is to build components from the edges determined in the following two stages:
\begin{itemize}
\item \textbf{Stage 1.} For each node $v$, chose an edge connecting $v$ to its neighbour with the largest degree. Then, determine the complete partition with respect to the set of chosen edges. Moreover, set the \emph{degree} of each component of the obtained partition as the maximum of the degrees of its elements.
\item \textbf{Stage 2.} For each node $v$, chose an edge connecting $v$ to a component $C\neq C(v)$ with the largest degree. Determine the complete partition with respect to the set of edges chosen in both stages. 
\end{itemize}
For $s\le n$, we say that components with degree at least $s$ are \emph{awake} components, while components with degree smaller than $s$ are \emph{asleep} components. Similarly, all nodes from awake components are called \emph{awake nodes} and nodes from asleep components are called \emph{asleep nodes}. As we show below, the complete partition $\mathbb{C}$ determined by the edges chosen in Stages~1 and 2 satisfies the following conditions:
\begin{enumerate}
\item[a)] the size of each awake component is larger than $s$,
\item[b)] each asleep node of $\mathbb{C}$ does not have neighbors with the degree larger than $s$,
\end{enumerate}

are satisfied for each $s\le n$. Algorithm~\ref{a:degreereduction} contains a pseudo-code of an implementation of the above described idea in the congested clique model in $O(1)$ rounds, in accordance with requirements of Lemma~\ref{l:degreereduction}. The above properties a)--b) combined with Alg.~\ref{a:degreereduction} imply Lemma~\ref{l:degreereduction}. Below, we provide a formal proof of Lemma~\ref{l:degreereduction} based on the above described ideas.

\begin{algorithm}
\caption{ReduceDegree$(v,s)$\Comment{execution at a node $v$, for $s\in\NAT$}}
\label{a:degreereduction}
\begin{algorithmic}[1]
\State coordinator$ \gets u_1$\Comment{set the coordinator} \label{ag:s1} 
\State $v$ announces $d(v)$ to all nodes in $N(v)$ \label{ag:s2} 
\Statex \textbf{Stage~1} 
\State $v$ sends the edge $(u,v)$ to the coordinator, where $u$ is the node with the largest ID from the set of neighbors of $v$ with highest degree, i.e., $\{w\in N(v)\, |\, d(w) = \max\limits_{t\in N(v)} d(t)\} $ \label{ag:s3} 
\State the coordinator calculates the complete partition determined by the received edges, sends the message $(I(v), d(C(v)))$ to each $v$ \label{ag:s4} 
\State $v$ announces the received message $(I(v), d(C(v)))$ to all nodes in $N(v)$ \label{ag:s5} 
\Statex \textbf{Stage~2}
\State $v$ sends the edge $(u,v)$ to the coordinator, where 
$C(u)$ is the highest degree component incident to $v$, i.e., $C(v) \neq C(u)$ and $d(C(u))$ is maximal among components incident to $v$ 
\label{ag:s6} 
\State the coordinator calculates the complete partition determined by all received edges (i.e., in both stages), sends the message $(I(u), d(C(u)))$ to each $u$ \label{ag:s7} 
\State $v$ announces the received message $(I(v), d(C(v)))$ to nodes in $N(v)$ \label{ag:s8} 
\State \textbf{if} $d(C(v))\ge s$ \textbf{then} $v$ is awake \textbf{else} $v$ is asleep
\State $G_A\gets (V, E_A)$, where $E_A=\{ (u,v)\in E\,|\, u,v \text{ are awake}\}$
\State $\mathbb{C}_A\gets$ the partition consisting from awake components \label{ag:s11} \Comment{$C$ is awake iff $v$ awake for some $v\in C$}
\State $G_B\gets (V,E_B)$, where $E_B=\{ (u,v)\in E\,|\, u\text{ or }v \text{ is asleep}\}$
\Comment{$E_B=E\setminus E_A$}
\end{algorithmic}
\end{algorithm}

\begin{fact}
\labell{f:reductiontosparse}
The following conditions are satisfied at the end of an execution of Alg.~\ref{a:degreereduction}: (i)~there are at most  $n/s$ awake components; (ii)~the degree of the graph $G_B$ induced by edges incident to the asleep nodes is smaller than $s$.
\end{fact}
\begin{proof}
Let $\mathbb{C}$ be a partition of the input graph obtained from edges announced in Stages~1 and 2. Let $\prec$ be the lexicographic ordering of the pairs $(d(v), \text{ID}(v))$ for $v\in V$.

Firstly, we show that each awake component $C$ of $\mathbb{C}$ has at least $s+1$ nodes, which implies~(i). For an awake component $C$, let $v_{\text{max}}\in C$ be the element of $C$ corresponding to the largest tuple in the set $\{(d(v),\text{ID}(v))\,|\, v\in C\}$. Thus, $d(v_{\text{max}})\geq s$. \kn{Moreover, each $u\in N(v_{\text{max}})$ belongs to $C(v_{\max})$ after Stage~1 of the algorithm. Indeed, if we contrary assume that 
\begin{itemize}
\item $d(v_{\text{max}})<s$: 

Then $d(v)<s$ for each $v\in C$ and therefore $d(C)<s$ and $C$ is asleep. This contradicts the assumption that $C$ is awake.

\item some neighbour $u$ of $v_{\max}$ does not belong to $C(v_{\max})$ after Stage~1:

Then, let $U$ be the set neighbours of $v_{\max}$ which are not in $C(v_{\max})$ after Stage~1. In Stage~2, $v_{\max}$ sends an edge connecting it with some $u\in U$. According to the algorithm, each $u\in U$ sends an edge $(u,w)$ in Stage~1 such that $(d(v_{\text{max}}),\text{ID}(v_{\text{max}}))\prec(d(w),\text{ID}(w))$. Thus $v_{\max}$ and $w$ as above are in the same component of the partition obtained after Stage~2, which contradicts the choice of $v_{\max}$.
\end{itemize}
}
Given that $d(v_{\text{max}})\geq s$, $v_{\text{max}}\in C$ and $N(v_{\text{max}})\subseteq C$, we see that the size of $C$ is larger than $s$.

Now, we prove the property (ii). As the degree of all asleep nodes is smaller than $s$, it is sufficient to show that the degrees of all neighbors of asleep nodes are smaller than $s$ as well. Contrary, assume that a node $u$ is asleep, $v\in N(u)$, and $d(v)\ge s$. Then, $u$ reports a node $w$ in Stage~1 such that $(d(v),\text{ID}(v))\preceq (d(w),\text{ID}(w))$ which implies that $s\le d(v)\le d(w)$. This in turn implies that $u$ and $w$ are eventually in the same component and, by the fact that $d(w)\ge s$, they are both awake. This however contradicts the assumption that $u$ is asleep.
\end{proof}

Now, we apply Fact~\ref{f:reductiontosparse} for $s=\log\log n$ to prove Lemma~\ref{l:degreereduction}. Let $G_A$ be the subgraph of $G$ containing edges whose both ends are awake. By Fact~\ref{f:reductiontosparse}(i), the partition determined in the algorithm contains at most $n/s=O(n/\log\log n)$ awake components. That is, the partition $\mathbb{C}_A$ of $G_A$ has $O(n/\log\log n)$ active components (see step~\ref{ag:s11} of Alg.~\ref{a:degreereduction}). And, the partition $\mathbb{C}_A$ is known to the coordinator. As the nodes learn components' IDs and degrees of their neighbours in step~\ref{ag:s8} of Alg.~\ref{a:degreereduction}, they know which edges incident to them belong to $G_A$ and which to $G_B$. The graph $G_B$ contains the edges incident to asleep nodes. By Fact~\ref{f:reductiontosparse}(ii), the degree of $G_B$ is smaller than $s=\log\log n$.

\subsection{Size-reduction Phase}\labell{s:componentreduction}

In this section we provide $O(1)$ round algorithm reducing the number of active components for sparse graphs (Alg.~\ref{a:componentreduction}). Assuming that the degree of the input graph $G(V,E)$ is at most $\Delta \in O(\log \log n)$, our algorithm returns a partition of the input graph into $O\left(n / \log \log n \right)$ active components. Additionally, some fixed node (the coordinator) knows a spanning tree of each component in the final partition. Thus, Lemma~\ref{l:componentreduction} follows from the properties of the presented algorithm.

In Algorithm~\ref{a:componentreduction}, $C_i(u_j)$ denotes the component of the node $u_j$ in the $i$th sample graph $G_i$. Moreover, for a node $u_i\in V$, let $e_{(i,1)},\ldots,e_{(i,r)}$ for $r\le |N(u_i)|$ denote all edges $(u_i,u_j)$ such that $j<i$.

The algorithm randomly selects $m=\sqrt{n}$ subgraphs $G_1,\ldots,G_m$ of the input graph $G$, called \emph{samples}. Each sample will consist of $O(n)$ edges, with high probability. We will ensure that all edges of the sample $G_i$ are known to a fixed node called the \emph{boss} $b_i$ (lines \ref{ss:send:edge:b}--\ref{ss:send:edge}). Therefore, for each sample, we can locally determine its connected  components and its spanning forest. Finally, the results from samples are combined in order to obtain a partition of the input graph which consists of $O(n/\log \log n)$ active components. The key challenge here is how to combine knowledge about locally available components of sample graphs such that significant progress towards establishing components of the original input graph is achieved. To this aim, we select randomly $\Theta(n/\log \log n)$ leaders among nodes of the input network. More precisely, each node of the network assigns itself the status \emph{leader} with probability $1/\log\log n$, independently of other nodes. Thus, the number of leaders is $\Theta(n/\log\log n)$, with high probability. Then, the idea is to build a (global) knowledge about connected components of the input graph by assigning nodes to the leaders which appear together with them in connected components of samples. \kn{In order to determine connected components without their spanning trees, it is sufficient to apply the following procedure:} if the connected component of $u_j$ in the $i$th sample contains some leader, the boss $b_i$ will send a message to $u_j$ containing the ID of that leader. \kn{However, as we want to determine spanning trees as well, we need a more complicated approach (see lines \ref{ss:compred:leader:b}--\ref{ss:compred:leader:e} of Alg.~\ref{a:componentreduction}):}
\begin{itemize}
\item If the connected component of $u_j$ in the $i$th sample contains a leader, the boss $b_i$ determines a shortest path $P$ connecting $u_j$ and a leader in $C_i(u_j)$.
\item If the connected component of $u_j$ in the $i$th sample does \textit{not} contain a leader, the boss $b_i$ determines a shortest  path $P$ connecting $u_j$ and the node of $C_i(u_j)$ with the smallest ID.
\item
\kn{Then, $b_i$ sends a message to $u_j$ containing the ID of the neighbour of $u_j$ in $P$. The message sent to $u_j$ contains some additional information which we need in order to deal with nodes which are connected to different leaders in various samples and nodes which are not connected to any leader. (Details are explained in proofs of Prop.~\ref{p:leader:connect} and Prop.~\ref{p:small:uncovered}.)} 
\end{itemize}

We will say that a component $C$ is \emph{small} if $C$ is ungrowable 
and the size of $C$ is at most $s=\sqrt{\log n}$.\kn{\footnote{Our results should work
for smaller $s$, e.g., polynomial wrt $\log\log n$. However, as it does not affect
round complexity of the algorithm, we choose $s$ which makes analysis easier.}}
In the following, we split nodes of the input graph into three subsets:
\begin{itemize}
\item
$V_{\alpha}$: the nodes connected to a leader in at least one sample graph,
\item
$V_{\beta}$: the elements of small components of the input graph, \kn{which do not
belong to $V_{\alpha}$,}
\item
$V_{\gamma}$: the remaining nodes of the graph; thus, $v$ belongs to $V_\gamma$
when $v$ is not an element of a small component of the input graph and
there are no leaders in connected components of $v$ in samples
$G_1,\ldots,G_m$.
\end{itemize}
In the analysis of Alg.~\ref{a:componentreduction}, we show
that each node from $V_\alpha$ will belong to a component containing a leader in the
final partition $\mathbb{C}$ determined by the coordinator.
\begin{proposition}\labell{p:leader:connect}
Assume that $C_i(v)$ (i.e., the connected component of $v\in V$ in $G_i$) 
for some $i\in[m]$ contains a leader.
Then, the connected component of $v$ in the final partition $\mathbb{C}$
contains a leader as well.
\end{proposition}
Moreover, we show that small components of the input graph are uncovered
by the coordinator with high probability, which determines the final components of nodes
from $V_\beta$.
\begin{proposition}\labell{p:small:uncovered}
The following property holds with high probability for 
each small component $C$ (i.e., an \kn{ungrowable} component of size at most $s=\sqrt{\log n}$) of the input graph:
$C$ is a connected component
of the final partition determined by the coordinator in Alg.~\ref{a:componentreduction}
or (at least one) leader belongs to $C$. 
\end{proposition}
While Prop.~\ref{p:leader:connect} and \ref{p:small:uncovered} concern $V_\alpha$ and $V_\beta$, 
we give an estimation of the size of $V_\gamma$ in the following proposition.
\begin{proposition}\labell{p:large:afew}
The number of nodes from $V_\gamma$ is $O(n/\log\log n)$, with
high probability.
\end{proposition}
We postpone the proofs of the above propositions and show properties
of Alg.~\ref{a:componentreduction} following from them. 
This in turn directly implies Lemma~\ref{l:componentreduction}.

\begin{algorithm}
\caption{Reduce Components in Sparse Graph}
\label{a:componentreduction}
\begin{algorithmic}[1]
\State $m\gets \sqrt{n}$
\State \textbf{for} $i\in[m]$ \textbf{do} $b_i\gets u_i$\Comment{the bosses}
\State the coordinator $\gets u_n$ \Comment{fix the coordinator}
\For{$i \in [n] $}\label{ss:send:edge:b}\Comment{simultaneously, in one round}
\For{$k \in [|N(u_i)|]$}
\For{$j \in [\sqrt{n}]$}
 \State $u_i$ adds $e_{(i,k)}$ to $G_j$ and sends it to $b_j$ with probability $1 / \log \log n$
		\label{ss:send:edge}\Comment{determine $G_j$}
\EndFor
\EndFor
\EndFor
\State \textbf{for} each $i\in[m]$ \textbf{do} $b_i$ calculates a spanning forest $F_i$ induced by received edges
\State each node, with probability $1 / \log \log n$ becomes a leader, and announces it to all bosses $b_i$
\For{$i \in [\sqrt{n}]$}
\For{$j \in [n] $}
\If{$C_i(u_j)$ contains a leader}\label{ss:compred:leader:b}\Comment{$C_i(u_j)$: the component of $u_j$ in $G_i$}
	\State \kn{$\textit{dist}\gets$ the length of a shortest path from $u_j$ to a leader in $C_i(u_j)$}
	\State $p_i(u_j)\gets $ the first node on a path from $u_j$ to the closest leader in $C_i(u_j)$
	\State $b_i$ sends the message $(n-\textit{dist}, |C_i(u_j)|, i, p_i(u_j))$ to $u_j$ \label{a:sr:largest1}
\Else 
	\State $p_i(u_j)\gets$ the first node on a shortest the path from $u_j$ to the node with the smallest ID in $C_i(u_j)$
	\State $b_i$ sends the message $(0, |C_i(u_j)|, i, p_i(u_j))$ to $u_j$\label{ss:compred:leader:e} 
\EndIf
\EndFor
\EndFor
\For{$j\in[n]$}
	\State let $(x,|C_i(u_j)|,i,p_i(u_j))$ be the largest message according to the 
	lexicographic order received by $u_j$\label{a:l:largest}
	\State $p(u_j)\gets p_i(u_j)$ 
	\State $u_j$ sends the edge $(u_j,p(u_j))$ to the coordinator\label{s:ujsends}
\EndFor
\State the coordinator computes components determined by the received edges
\end{algorithmic}
\end{algorithm}

\begin{lemma} 
\labell{l:componentreductionalgcorrectness}
At the end of Algorithm~\ref{a:componentreduction},
a partition with at most $O(n/\log \log n)$ active components is determined, 
the coordinator knows this partition and 
a spanning tree for each component of the partition.
\end{lemma}
\begin{proof}
Let $\mathbb{C}$ be the final partition determined by the coordinator.
Proposition~\ref{p:small:uncovered} implies that all small components of the input graph are also components
of $\mathbb{C}$. Thus, they are inactive in $\mathbb{C}$. As there are $\Theta(n/\log\log n)$ leaders
whp, Proposition~\ref{p:leader:connect} implies that all nodes from $V_\alpha$ belong to $O(n/\log\log n)$
active components. Finally, as there are only $O(n/\log\log n)$ nodes from $V_\gamma$, with high probability (Prop.~\ref{p:large:afew}), there are at most $O(n/\log\log n)$ components of $\mathbb{C}$ containing those nodes.

Finally, as the coordinator computes the final partition into connected components 
based on the received edges (step~\ref{s:ujsends}),
it can also determine spanning trees of the components of this partition.
\end{proof}

The remaining part of this section contains the proofs of Propositions~\ref{p:leader:connect}, \ref{p:small:uncovered}, and \ref{p:large:afew}.

\subsubsection{Connections to leaders: Proof of Prop.~\ref{p:leader:connect}}
\kn{Let $\text{dist}_{\text{leader}}(v)$ be the length of a shortest path
connecting a node $v$ with a leader in the sample graphs $G_1,\ldots,G_m$,
provided $v$ is connected with a leader in some sample.

We prove the proposition by induction with respect to the value of $\text{dist}_{\text{leader}}(v)$. The fact that $\text{dist}_{\text{leader}}(v)=0$ means that $v$ is a leader. Thus, $v$ certainly is in the connected component of the final partition $\mathbb{C}$ containing a leader. For the inductive step, assume that the proposition holds for each node $v$ such that $\text{dist}_{\text{leader}}(v)<j$ for some $j<n$. Let $v$ be a node connected to a leader in some sample such that  $\text{dist}_{\text{leader}}(v)=j$. Thus, a shortest path connecting $v$ and a leader in a sample has length $j$. Therefore, the largest tuple according to lexicographic ordering obtained from the bosses by $v$ is $(n-j, |C_i(v)|, i, p_i(v))$ for some $i\in[m]$, where $p_i(v)$ is a neighbour of $v$ \tj{in distance $j' < j$ from a leader}, i.e., $\text{dist}_{\text{leader}}(p_i(v))=j' < j$. Thus, 
\begin{itemize}
\item $p(v)$ will be assigned the value $p_i(v)$ in the algorithm and $v$ sends an edge $(v,p_i(v))$ to the coordinator,
\item as $\text{dist}_{\text{leader}}(p_i(v))= j' < j$, the inductive hypothesis guarantees that $p_i(v)$ is connected with a leader in the final partition $\mathbb{C}$ determined by the coordinator.
\end{itemize}
\tj{Therefore}, as $p_i(v)$ is connected with a leader in the partition $\mathbb{C}$ determined by the coordinator and the edge $(v,p_i(v))$ is also known to the coordinator, $v$ is connected with a leader in $\mathbb{C}$ as well.
}

\subsubsection{Spanning trees of small components: Proof of Prop.~\ref{p:small:uncovered}}

Before the formal proof of Prop.~\ref{p:small:uncovered}, we give a general statement regarding
connected subgraphs of size $s\le 3\sqrt{\log n}$ of the input graph $G$.
Below, we show that a spanning tree of $G'$ will appear in some sample, with high probability.
\begin{proposition}
\labell{p:stprobability}
For a given set of nodes $V'$ of size $s \leq 3\sqrt{\log n}$ such that 
the subgraph of $G$ induced by $V'$ is connected,
there is no spanning tree of $V'$ in all samples with probability at most $O\left(\frac{1}{n^{\omega\left(1\right)}}\right)$.
\end{proposition}
\begin{proof}
A spanning tree of $V'$ consists of at most $s-1$ edges. 
Thus, it is present in some particular random sample
with probability $\prob(\mathit{present}) \geq \left(\frac{1}{\log \log n}\right)^{s-1}$. 
Thus, with probability at most $1-\prob(\mathit{present})$,
 it is not present in some particular random sample, and is not present in all samples simultaneously with probability \\
 \begin{align*}
 \left(1-\prob \left( \mathit{present} \right) \right)^{\sqrt{n}} &\leq 
 \left(1-\left(\frac{1}{\log \log n}\right)^{s-1}\right)^{\sqrt{n}} = 
 \left(1-\left(\frac{1}{\log \log n}\right)^{s-1}\right)^{ \left( \log \log n \right)^{s-1} \sqrt{n} \left( \log \log n \right)^{1-s} }\\
 &\leq 
 \left(\frac{1}{e}\right)^{\frac{\sqrt{n}}{ \left( \log \log n \right)^{s-1}}} =
 O\left(\frac{1}{n^{\omega\left(1\right)}}\right)
 \end{align*}
\end{proof}
Using Prop.~\ref{p:stprobability}, we will prove Prop.~\ref{p:small:uncovered}.
Let $C$ be a (ungrowable) connected component of the input graph $G$ such that $|C| \leq \sqrt{\log n}$ \kn{and no leader belongs to $C$.} By Prop.~\ref{p:stprobability}, $C$ has no spanning tree in all samples with probability at most $O(\frac{1}{n^{\omega(1)}})$. As there are at most $n$ small components, by union bound, the probability 
that there exists a small component of the input graph which is not a component of any sample is at most
$$n \cdot O\left(\frac{1}{n^{\omega(1)}}\right) = O\left(\frac{1}{n^{\omega(1)}}\right).$$ 
Therefore, with probability $1 - O(\frac{1}{n^{\omega(1)}})$, each small component of the input graph
is a component of some sample.
Thus, in order to prove Prop.~\ref{p:small:uncovered}, it is sufficient to show the following fact: 
if a small component $C$ of the input graph
is a component of some sample and no leader belongs to $C$, 
then $C$ will also be a component of the partition
$\mathbb{C}$ determined by the coordinator.
Assume that $C$ is a small component of $G$, $C$ is a component of
a sample $G_i$ for $i\in [m]$ and no leader belongs to $C$. 
W.l.o.g. assume that $i$ is the largest index of a sample
containing $C$ as a component. 
That is, $C$ is not a component of $G_j$ for $j>i$ and $C$ is a component of $G_i$.
Let $v_{\min}$ be the node of $C$ with the smallest ID.
Then, for each $v\in C$, the boss $b_i$ sends $(0, |C|, i, p_i(v))$ to $v$,
where $p_i(v)$ is the parent of $v$ in a tree $T$ rooted at $v_{\min}$,
consisting of shortest paths 
between $v_{\min}$ and other elements of $C$.
The choice of $i$ and the assumption that there are no leaders in $C$ guarantee that, 
for each $v\in C$, 
the message received by $v$ from $b_i$ is the largest message 
according to the lexicographic
ordering 
among messages received by $v$
from the bosses (see line~\ref{a:l:largest}).
Thus, each $v\in C$ sends $p(v)=p_i(v)$ to the coordinator.
Thanks to that fact, the coordinator learns about the described above spanning tree $T$
of $C$. 
Hence, $C$ is a component of the partition determined by the coordinator.
%
%
Therefore, the coordinator knows a spanning tree for every small connected component $C$ of $G$
with probability at least $1 - O(\frac{1}{n^{\omega(1)}})$.

\subsubsection{Leaderless nodes in large components: Proof of Prop.~\ref{p:large:afew}}
\labell{ss:leaderless}

We say that a node is \emph{bad} if it does not belong to a small component
nor to a component
containing a leader in the final partition $\mathbb{C}$ determined by
the coordinator. That is, a node is bad iff it belongs to $V_\gamma$.
In order to prove Prop.~\ref{p:large:afew}, it is sufficient to show that there are $O(\frac{n}{\log \log n})$ bad nodes.
Then, in the worst-case, bad nodes would be partitioned into $\Theta(\frac{n}{\log \log n})$ active components. 

The outline of the proof is as follows. 
Firstly, we cover all nodes from non-small components (i.e., from components of size at least $s=\sqrt{\log n}$) of the input graph
by connected sets $V_1, V_2,\ldots, V_r$ of sizes in the range $[s, 3s]$ for $s=\sqrt{\log n}$ (Fact~\ref{f:treepartition}) such that $V_i$'s are ``almost pairwise disjoint'' (a more precise definition will be provided later). 
Then, we associate the random $0/1$ variable $X_i$ to each set $V_i$
such that $X_i=0$ implies that no nodes from $V_i$ are bad. 
(In particular, $X_i=0$ holds when $V_i$ contains a leader and at least
one sample graph $G_j$ contains a spanning tree of $V_i$. Thus, by Prop.~\ref{p:leader:connect}, the nodes of $V_i$ are not bad if $X_i=0$.)
Importantly,
the variables $X_i$ are independent and the probabilities $\prob(X_i=1)$
are small. 
As the number of bad nodes is at most
$$\sum_{i\in[r]} |V_i|X_i\le 3\sqrt{\log n}\sum_{i\in[r]} X_i,$$
we prove the upper bound on $\sum_i X_i$ which ensures
that the number of bad nodes is $O(n/\log\log n)$
with high probability.

We start with a cover of non-small components by connected ``almost pairwise disjoint'' components of sizes
in the range $[s,3s]$, called an \emph{almost-partition}.
More precisely, we say that sets $A_1,\ldots,A_k$ form an
\emph{almost-partition} of a set $A$ iff $\bigcup_{i=1}^k A_i=A$
and, for each $j\in[k]$, $A_j$ contains at most one element belonging to other sets from
$A_1,\ldots, A_k$, i.e., $|A_j\cap\bigcup_{i\neq j}A_i|\le 1$.
The elements of $A_i\setminus \bigcup_{i\neq j}A_i$ are called \emph{unique}
for $A_i$.
\begin{fact}
\labell{f:treepartition}
Let $T$ be a tree of size at least $s\in\NAT$.
Then, there exists 
an almost-partition $T_1, T_2, \dots, T_k$ of $T$
such that 
\begin{equation}\label{e:almost}
T_i \mbox{ is a connected subgraph of } T \mbox{ and }
|T_i| \in [s, 3s]\mbox{ for each }i\in[k]. 
\end{equation}
\end{fact}
\begin{proof}
We prove the statement of the fact inductively.
If the size of $T$ is in the interval $[s,3s]$, the
almost-partition consisting from $T$ only satisfies the given
constraints.

For the inductive step, assume that the fact holds for trees
of size at most $n$ for $n>3s$.
Let $T$ be a tree of size $|T|=n+1$ on a set of nodes $V$.
In the following, we say that a subgraph $T'$ of $T$ 
induced by $V'\subseteq V$
is a \emph{subtree} of $T$ iff $T'$ and the subgraph of $T$ induced
by $V\setminus V'$ are trees.
Assume that $T$ contains a subtree $T'$ such that $|T'| \ge s$
and $|T|-|T'|\ge s$. Then, 
by the inductive hypothesis, there exists an almost-partition of $T\setminus T'$ and an 
(one element) almost-partition of $T'$ satisfying (\ref{e:almost}). Thus, an almost-partition of $T$ obtained from the almost-partitions of $T\setminus T'$ and of $T'$
%
satisfies (\ref{e:almost}) as well.

Now, assume that 
\begin{equation}\label{e:balans}
T \mbox{ does not contain a subtree } T' \mbox{ such that } |T'| \ge s \mbox{ and }|T|-|T'|\ge s.
\end{equation} 
%
For a tree with a fixed root $r$, $T(u)$ denotes a subtree of $T$ rooted at $u$. Now, we show an auxiliary property of trees satisfying (\ref{e:balans}).
\begin{claim}\labell{cl:balans}
Let $T$ be a tree satisfying (\ref{e:balans}). Then, one can chose $r\in T$ as the root of $T$ such that
%
%
\begin{enumerate}
\item[(a)]
$|T(v_i)|<s$ for each $i\in[k]$, where $\{v_1,\ldots,v_k\}$ is the set
of children of $r$.
\end{enumerate}
\end{claim}
\noindent\textit{Proof of Claim~\ref{cl:balans}.}
Let $r_0$ be an arbitrary node of a tree $T$ which satisfies (\ref{e:balans}).
If (a) is satisfied for $r=r_0$, we are done. 
Otherwise, we define the sequence $r_0, r_1,\ldots$ of nodes such that 
$r_{i+1}$ for $i\ge 0$ is the child of $r_i$ in $T$ (rooted at $r$) with the largest subtree. Then,
\begin{itemize}
\item[(i)]
$|T(r_i)|>|T(r_{i+1})|$, because $T(r_{i+1})$ is a subtree of $T(r_i)$,
\item[(ii)]
if $|T(r_i)|\ge s$, then $|T\setminus T(r_i)|<s$, by the assumption (\ref{e:balans}).
\end{itemize}
The condition (i) guarantees that $|T(r_{j})|\ge s$ and
$|T(r_{j+1})|<s$ for some $j\ge 0$. 
Thus, $|T(v)|<s$ for all children of $r_{j}$,
since $T(r_{j+1})$ has the largest size among subtrees rooted at children of $r_j$.
The assumption (\ref{e:balans}) implies also
that $|T\setminus T(r_j)|<s$. Thus, (a) is satisfied for $T$ if 
the root $r$ is equal to $r_j$. (\textit{Proof of Claim~\ref{cl:balans}})\qed

Using Claim~\ref{cl:balans}, we can choose the root $r$ of $T$ such that
$$s>|T(v_1)|\ge |T(v_2)|\ge\cdots\ge |T(v_k)|,$$
where $\{v_1,\ldots,v_k\}$ is the set of children of $r$ in $T$ (when
$r$ is the root of $T$).
Next, we split the set of trees $T(v_1),\ldots,T(v_k)$
into subsets such that the number of nodes in each subset is in
the range $[s-1,3s-1]$. 
Such a splitting is possible thanks to the facts that $|T(v_i)|<s$
for each $i\in[k]$ and $\sum_{i=1}^k|T(v_i)|\ge 3s$.
Finally, by adding the node $r$ to each subset,
we obtain an almost-partition satisfying (\ref{e:almost}).

\end{proof}

Using Fact~\ref{f:treepartition}, we will eventually prove Prop.~\ref{p:large:afew}.
Let $S_1, S_2, \dots, S_k$ be non-small connected components of the input graph $G$, 
i.e., $|S_i| > \sqrt{\log n}$ for each $i\in[k]$. By Fact~\ref{f:treepartition}, there exists an almost-partition of spanning trees of $S_i$'s into trees of sizes from the interval $[\sqrt{\log n}, 3\sqrt{\log n}]$. 
Let $\mathbb{T}=\{T_1,T_2,\ldots\}$ be the set of trees equal to the union of all those almost-partitions.
Observe that, according to the properties of almost-partitions, there are at least $\sqrt{\log n} - 1$ nodes \emph{unique} for $T_{i}$ for each $i$, i.e.,
nodes which belong to $T_{i}$ and do not belong to any other tree of
the above specified almost-partitions of the components $S_1, S_2, \dots, S_k$.
%
Thus, $|T|=O(n/\sqrt{\log n})$.
We associate random events $A_{i}$ and $B_{i}$ with each tree 
$T_{i}$, where
\begin{itemize}
\item
$A_{i}$ is the event that all edges of $T_{i}$ appear
in at least one sample graph among $G_1,\ldots,G_m$ in an execution of 
Alg.~\ref{a:componentreduction};
\item
$B_{i}$ is the event that at least one 
element of the set of nodes unique for $T_{i}$
has the status leader in an execution of 
Alg.~\ref{a:componentreduction}.
\end{itemize}
Importantly, all event $A_{i}$ and $B_{i}$ are independent,
thanks to the facts that each edge is decided to be included in
each sample graph independently, the sets of edges of $T_{i}$'s are 
disjoint, the set of nodes \textbf{unique} for $T_{i}$'s are disjoint 
as well, and the random choices determining whether a node has
a status leader are also independent.

\par By Prop.~\ref{p:stprobability} 
the probability of $A_{i}$ is $\prob(A_{i})=1-O(\frac{1}{n^{\omega(1)}})$. 
As $T_{i}$ has at least $s=\sqrt{\log n}$ unique nodes,
the probability of $B_{i}$ is at least
$$\prob(B_{i})\ge 1-\left(1-\frac{1}{\log \log n}\right)^{\sqrt{\log n}-1}.$$
Observe that the conjunction of the events $A_{i}$ and $B_{i}$
guarantees that the nodes of $T_{i}$ are connected to a leader
in the partition $\mathbb{C}$, i.e., they are not bad nodes.
Let $X_{i}$ be a 0/1 variable, where $X_{i}=0$ iff $A_{i}$ and 
$B_{i}$ are satisfied. 
Thus, $X_{i}=0$ implies that no node from $T_{i}$ is bad.
As the events $A_{i}$ and $B_{i}$ are independent, the probability
that $X_{i}=0$ (implying that no node from $T_{i}$ is bad)
can be estimated as follows:
\begin{alignat*}{2}
\prob(X_{i}=0) & > \prob(A_i)\cdot \prob(B_i)\\
& > 
\left(1 -O\left(\frac{1}{n^{\omega\left(1\right)}}\right)\right) \cdot \left(1-\left(1-\frac{1}{\log \log n}\right)^{\sqrt{\log n}-1}\right) \\
& > 
\left(1 -O\left(\frac{1}{n^{\omega\left(1\right)}}\right)\right)\cdot \left(1 -\frac{1}{e^{(\sqrt{\log n}-1)/ {\log \log n}}}\right) \\
& = 1 -O\left(\frac{1}{e^{(\sqrt{\log n}-1)/ {\log \log n}}}\right).
\end{alignat*}

\comment{
Let consider random variables $X_{j,i}$, which takes value $0$, when in some sample there exists spanning tree of $T_{S_j,i}$, and among nodes unique for tree $T_{S_j,i}$ there is a leader. Thus, variable $X_{j,i}$ takes value $1$ with probability at most $ O\left( \frac{1}{n^{\omega(1)}}\right) + 1 - \left( 1 -O\left(\frac{1}{e^{(\sqrt{\log n}-1)/ {\log \log n}}}\right) \right) = O\left(\frac{1}{e^{(\sqrt{\log n}-1)/ {\log \log n}}}\right)$
}
Thus, $\prob(X_{i}=1) =O\left(\frac{1}{e^{(\sqrt{\log n}-1)/ {\log \log n}}}\right)$.
The expected number of \textit{bad} nodes 
is upper bounded by
\begin{equation}\label{e:xi}
E\left[\sum\limits_i X_{i}\cdot |T_{i}|\right]
= O\left(\sqrt{\log n}\right) \sum\limits_i X_{i},
\end{equation}
since $|T_{i}| \in \Theta(\sqrt{\log n})$ for each $i$ under
consideration. 
%
The expected value of the sum of the variables $X_{i}$ can be estimated
as
\begin{alignat*}{2}
E\left[\sum\limits_i X_{i}\right] & =\sum\limits_i O\left(\frac{1}{e^{(\sqrt{\log n}-1)/ {\log \log n}}}\right)
= O\left(\frac{n}{\sqrt{\log n}}\right) O\left(\frac{1}{e^{(\sqrt{\log n}-1)/ {\log \log n}}}\right) \\
&= O\left(\frac{n}{\sqrt{\log n}\cdot e^{(\sqrt{\log n}-1)/ {\log \log n}}}\right).
\end{alignat*}
As $X_{i}$ are independent $0-1$ random variables, 
\begin{equation}\label{e:chern}
\sum\limits_{i} X_{i} \in O\left(\frac{n}{\sqrt{\log n}\cdot e^{(\sqrt{\log n}-1)/ {\log \log n}}}\right)
\mbox{ with high probability}
\end{equation}
by a standard Chernoff bound.
Therefore, by (\ref{e:xi}) and (\ref{e:chern}), the number of \textit{bad} nodes is 
$$O(\sqrt{\log n})\cdot O\left(\frac{n}{\sqrt{\log n}\cdot e^{(\sqrt{\log n}-1)/ \log \log n}}\right) = O\left(\frac{n}{e^{(\sqrt{\log n}-1)/ \log \log n}}\right) = O\left(\frac{n}{\log \log n}\right)$$ 
with high probability.
%
This fact finishes the proof of Prop.~\ref{p:large:afew}.

	\section{\MST{} in $O(1)$ rounds}\labell{s:MST}

\par In order to find \MST{} of a given input graph, we will use 
the $O(1)$ round \CC{} algorithm and the reduction from 
\cite{HegemanPPSS15}. 
The \MST{} problem for a given graph can be reduced, by using the KKT random sampling \cite{Karger:1995:RLA:201019.201022} to two consecutive instances of \MST{}, each for a graph with $O(n^{3/2})$ edges. 
The authors of \cite{HegemanPPSS15} observed that the \MST{} problem for a graph with $O(n^{3/2})$ edges can be reduced to $\sqrt{n}$ instances of the \CC{} problem simultaneously. In each of those $\sqrt{n}$ instances of the \CC{} problem, the set of neighbours of each node is a subset of the set of its neighbours in the original input graph with $O(n^{3/2})$ edges (a nice exposition of the reduction is also given in \cite{GhaffariParter2016}). 
For further references, we state these reductions more precisely.
\begin{lemma}\cite{HegemanPPSS15}\labell{l:red1:hegeman}
Let $G(V,E,c)$ be an instance of the \MST{} problem. 
There are congested clique $O(1)$ rounds algorithms $A_1, A_2$ such 
that, with high probability,
\begin{enumerate}
\item
$A_1$ builds $G_1(V,E_1,c)$ with $O(n^{3/2})$ edges such that
$E_1\subseteq E$ and,
\item
given a minimum spanning forest of $G_1$, $A_2$ builds 
$G_2(V,E_2,c)$ with $O(n^{3/2})$ edges such that 
$E_2\subseteq E$ and
a minimum spanning tree of $G_2$ is also
a minimum spanning tree of $G$.
\end{enumerate}
\end{lemma}

\begin{lemma}\cite{HegemanPPSS15}\labell{l:red2:hegeman}
Let $G(V,E,c)$ be an instance of \MST{}, where $|E|=O(n^{3/2})$.
There is a congested clique $O(1)$ round algorithm which
reduces the \MST{} problem for $G$ to $m=\sqrt{n}$ instances 
$G_i(V,E_i)$ for $i\in[m]$ of
the \CC{} problem, such that (i)~$E_1\subseteq E_2\subseteq\cdots\subseteq E_m=E$;
(ii)~each node $v$ knows edges incident to $v$ in $E_i$
for each $i\in[m]$ at the end of an execution of the algorithm.
\end{lemma}

If we show that our \CC{} algorithm can be executed 
simultaneously for $\sqrt{n}$ instances satisfying the properties
from Lemma~\ref{l:red2:hegeman} in $O(1)$ rounds, then we obtain the $O(1)$ round
randomized algorithm for
\MST{}.
\par As shown in Section~3 of \cite{GhaffariParter2016}, the algorithm ReduceCC 
applied in GPReduction (Lemma~\ref{l:GP2016}) can be executed in parallel for $\sqrt{n}$ instances as above. 
As the algorithm GPReduction satisfying Lemma~\ref{l:GP2016} consist of $O(1)$ executions of ReduceCC
(see the proof of Lemma~\ref{l:GP2016}), this algorithm can be executed in parallel 
for $\sqrt{n}$ instances of the \CC{} problem satisfying conditions from Lemma~\ref{l:red2:hegeman}.
\par Therefore, in order to prove that our \CC{} algorithm can be executed in parallel for $\sqrt{n}$ instances described in Lemma~\ref{l:red2:hegeman}, it is sufficient to show that \kn{executions of Algorithm~\ref{a:degreereduction} and Algorithm~\ref{a:componentreduction} called in Alg.~\ref{a:cc} can be executed in parallel for such instances.} 
In Sections~\ref{ss:parallel1} and \ref{ss:parallel2}, we show that it is the case. 
Finally, in Section~\ref{ss:parallel:final}, we discuss a parallel execution of the whole
Alg.~\ref{a:cc} for all those instances.
This gives a $O(1)$ round
randomized algorithm determining a minimum spanning tree which proves
Theorem~\ref{t:MST}.
\comment{
\begin{theorem}\labell{t:MST}
There is a randomized algorithm in the congested clique model that computes
MST in $O(1)$ rounds, with high probability.
\end{theorem}
}

\subsection{Parallel executions of Alg.~\ref{a:degreereduction}}\labell{ss:parallel1}

In this section we show that Alg.~\ref{a:degreereduction} can be executed
for $\sqrt{n}$ related sparse instances of the \CC{} problem in parallel, 
as stated in the following lemma.
\begin{lemma}\labell{l:parallel1}
Let $G_1(V,E_1), \ldots, G_m(V,E_m)$ for $m=\sqrt{n}$ 
be the input graphs in the 
congested clique model such that $|E_i|=O(n^{3/2})$ for each
$i\in[m]$, $E_1\subseteq E_2\subseteq\cdots\subseteq E_m$,
and $v_j$ knows its neighbours in each of the graphs
$G_1,\ldots, G_m$.
Then, Alg.~\ref{a:degreereduction} can be executed simultaneously
for $G_1,\ldots,G_m$ in $O(1)$ rounds in the following framework:
\begin{itemize}
\item
for each $i\in[m]$, $j\in[n]$, the node $u_j$ has assigned a node
$\text{proxy}(i,j)\in\{u_1,\ldots,u_n\}$ such that $\text{proxy}(i,j)$
works on behalf of $u_j$ in the $i$th instance of the \CC{} problem;
\item
for each $j\in[n]$, the node $u_j$ works \tj{as the proxy on behalf of $O(1)$ 
nodes (possibly in many instances of the \CC{} problem),} i.e., 
\kn{$$|\{k\,|\,u_j=\text{proxy}(i,k)\}|=O(1).$$}
\end{itemize}
\end{lemma}
\noindent The proof of Lemma~\ref{l:parallel1} is presented in the remaining part of this section.
As there are $m$ instances of the \CC{} problem to solve (and therefore
$m$ instance of Alg.~\ref{a:degreereduction}), we can set
$m$ coordinators $c_1,\ldots,c_m$ such that, e.g., the node $u_i$ acts
as the coordinator $c_i$.

\par Steps \ref{ag:s3}, \ref{ag:s4}, \ref{ag:s6} and \ref{ag:s7}
of Alg.~\ref{a:degreereduction} either do not require any communication (steps \ref{ag:s4} and \ref{ag:s7}) or each node transmits a single message to the coordinator (steps \ref{ag:s3} and \ref{ag:s6}). Thus, these steps can be executed in parallel in
$O(1)$ rounds: instead of sending a message to one coordinator, each node
can send appropriate messages to the coordinators $c_1,\ldots,c_m$ in 
a round for $m=\sqrt{n}$.

\comment{
$2,3,5,6$, of Alg.it is communication with the boss. For $\sqrt{n}$ simultaneous instances we will have $\sqrt{n}$ separate bosses. Parallel execution does not influence this part of the algorithm - each leader still will send/receive $O(n)$ messages per round. As any node participates in at most $\sqrt{n}$ different instances, in each instance it simulates $O(1)$ nodes the number of messages sent/received by all nodes while communicating with bosses is $\Theta(\sqrt{n})$, thus again it can be done by Lenzen's routing \cite{Lenzen:2013:ODR:2484239.2501983}
(Lemma~\ref{l:lenzen}).
}
The main problem with parallel execution of 
of steps \ref{ag:s2}, \ref{ag:s5} and \ref{ag:s8} of Alg.~\ref{a:degreereduction}
is that each node sends a message to all its neighbours. 
%
%
Thus, a node with degree $\Delta=\omega(\sqrt{n})$ needs to send 
$$m\cdot \Delta=\sqrt{n}\cdot \omega(\sqrt{n})=\omega(n)$$
messages in $m$ parallel executions of Alg.~\ref{a:degreereduction}, which cannot be done in $O(1)$ rounds, even with help of e.g. Lenzen's routing, because of limited bandwidth of edges.

In order to overcome the above observed problem, we will take advantage
of the fact that the number of edges in each of $m=\sqrt{n}$ instances
of the problem is $O(n^{3/2})$. Thus, the overall number of messages to send 
in all instances is 
$$O\left(\sum_{v\in V}\sum_{i\in[m]} d_{G_i}(v)\right)=O\left(\sum_{i\in[m]}\sum_{v\in V} d_{G_i}(v)\right)=O\left(\sum_{i\in[m]} |E_i|\right)=O(\sqrt{n}\cdot n^{3/2})=O(n^2).$$
Hence, the amount of communication fits into quadratic number of edges
of the (congested) clique.

In our solution, we distribute communication load among so-called \emph{proxies}.
Let 
$d(u_j)=|\{(u_j,v)\,|\, (u_j,v)\in E_m\}|$ 
be the upper bound on the degree of $u_j$
in all graphs $G_1,\ldots,G_m$, due to the assumption $E_1\subseteq\cdots\subseteq E_m$.
Assume that a pool of \emph{proxy} nodes
$v_1,v_2,\ldots$ is available (see Lemma~\ref{l:auxiliary}).
We assign $l_j=\lceil d(u_j)/\sqrt{n}\rceil$ proxy nodes to 
$u_j$ for each $j\in[n]$. Altogether, we need
$$\sum_{j\in[n]}\lceil d(u_j)/\sqrt{n}\rceil\le n+\frac1{\sqrt{n}}\sum_{j\in[n]}d(u_j)=n+O\left(\frac1{\sqrt{n}}\cdot n^{3/2}\right)=O(n)$$
proxy nodes. 
The key idea is that the work of the node $u_j$ 
is split between its proxies such that each proxy node is responsible
for simulating $u_j$ in $\min\{\sqrt{n},\lceil n/d(u_j)\rceil\}$ instances of 
Alg.~\ref{a:degreereduction}.
In order to guarantee feasibility of a simulation of all $\sqrt{n}$
executions of Alg.~\ref{a:degreereduction} by the proxies in $O(1)$ rounds, we have to address
the following issues:
\begin{enumerate}

\item[(a)]
%
In order to simulate the nodes $\{u_1,\ldots,u_n\}$ in 
various instances of the \CC{} problem,
the proxies need to know the mapping
between the nodes $\{u_1,\ldots,u_n\}$ and proxies
simulating them in respective instances of the \CC{} problem.

\item[(b)]
In order to simulate $u_j$ in the $i$th instance of the problem,
the appropriate proxy should know the neighbors of $u_j$ in $G_i$. Thus,
information about neighbours of appropriate nodes should be delivered to proxies 
before the actual executions of Alg.~\ref{a:degreereduction}.

\item[(c)]
It should be possible that each proxy $v$ of $u_j$ is able to simulate 
each step of $u_j$ in all instances of the problem in which $v$ works
on behalf of $u_j$ in $O(1)$ rounds.


\end{enumerate}

Regarding (a), note that the values $d(u_j)$ can be distributed to all nodes in a single round. Using this information, each node can locally compute which proxies are assigned to particular nodes of the network in consecutive instances of the problem, assuming that the proxies are assigned in the ascending order, i.e., $v_1,\ldots, v_{l_1}$ are assigned to $u_1$, $v_{l_1+1},\ldots, v_{l_1+l_2}$ are assigned to $u_2$ and so on.

As for (b), we rely on the fact that $E_1\subseteq\cdots\subseteq E_m$. The node $u_j$ encodes information about its neighbours in all sets $E_i$ in a following way: for each edge $e$ it is enough to remember what is the smallest $i$, such that $e \in E_i$. More formally,  $u_j$ encodes it as the set 
\kn{$$T_j=\{(v,l)\,|\, (u_j,v)\in E_{l'}\text{ for }l'\ge l, (u_j,v)\not\in E_{l'}\text{ for }l'<l\}.$$
That is, $(v,l)\in T_j$ iff $(u_j,v)\in E_l, E_{l+1},\ldots,E_m$, $(u_j,v)\not\in E_1,\ldots,E_{l-1}$.} Thus, knowing $T_j$, it is possible to determine the neighbors of $u_j$ in $G_i$, for each $i\in[m]$. The set $T_j$ is delivered to all $\lceil d(u_j)/\sqrt{n}\rceil$ proxies of $u_j$ in the following way:
\begin{itemize}
\item \textbf{Stage~1.}
The set $T_j$ is split into $\lceil d(u_j)/\sqrt{n}\rceil$ subsets
of size $\sqrt{n}$, each subset is delivered to a different proxy
of $u_j$.
\item \textbf{Stage~2.}
The subset of $T_j$ of size $\sqrt{n}$ delivered to a proxy of $u_j$ 
in Stage~1 is delivered to all other proxies of $u_j$.
\end{itemize}
In order to perform Stages~1 and 2 in $O(1)$ rounds, we apply
Lenzen's routing algorithm \cite{Lenzen:2013:ODR:2484239.2501983} (Lemma~\ref{l:lenzen}),
which works in $O(1)$ rounds provided each node has $O(n)$ messages
to send and $O(n)$ messages to receive.
Note that $u_j$ has $d(u_j)=O(n)$ messages to be transmitted
and each proxy has $O(\sqrt{n})=O(n)$ messages to receive in Stage~1.
In Stage~2, each proxy of $u_j$ is supposed to deliver and receive
$$O\left(\sqrt{n}\cdot\frac{d(u_j)}{\sqrt{n}}\right)=O(n)$$
messages.

Knowing that the above issues (a) and (b) are resolved, we can also address (c). 
Thanks to the presented solution for (a), there is global
knowledge which proxy nodes are responsible for particular nodes of the
network in various instances of the problem and all proxies start with
the knowledge of the nodes simulated by them.
Thus, in each step of the algorithm, transmissions between nodes
are replaced with transmissions between appropriate proxy nodes.
It remains to verify whether proxy nodes are able to deliver messages
on behalf of the actual nodes simulated by them in \textbf{all}
instances of Alg.~\ref{a:degreereduction} simulated by them.
The number of messages supposed to be sent and received by the node 
$u_j$ in a step of of Alg.~\ref{a:degreereduction} is at most $d(u_j)$.
As each proxy of $u_j$ simulates $u_j$ in 
$\min\{\sqrt{n},\lceil n/d(u_j)\rceil\}$ instances of Alg.~\ref{a:degreereduction}, it is supposed to send/receive at most
$$O\left(d(u_j)\cdot \min\{\sqrt{n},\lceil n/d(u_j)\rceil\}\right)=O(n)$$
messages in each round.
Hence, using Lenzen's routing \cite{Lenzen:2013:ODR:2484239.2501983} (Lemma~\ref{l:lenzen}), each step of each execution
might be simulated in $O(1)$ rounds.

\comment{
Regarding (e), recall that we obtain two graphs $G_1, G_2$
at the end of Algorithm~\ref{a:degreereduction}, where $G_1$ 
has $O(n/\log\log n)$ active components and the degree of $G_2$
is $O(\log\log n)$.
}

\subsection{Parallel executions of Algorithm~\ref{a:componentreduction}}\labell{ss:parallel2}

In this section, we show feasibility of simulation of $\sqrt{n}$ instances of
Alg.~\ref{a:componentreduction} in $O(1)$ rounds.
\begin{lemma}\labell{l:parallel2}
Assume that $m=\sqrt{n}$ graphs of degree $O(\log\log n)$ are
given in the congested clique model, i.e., each node knows its
neighbors in each of the graphs.
Then, Alg.~\ref{a:componentreduction} can be executed
simultaneously on all those instances in $O(1)$ rounds.
\end{lemma}

In order to prove Lemma~\ref{l:parallel2},
it is sufficient to analyze the number of messages
which nodes send/receive in the case that they 
simulate $\sqrt{n}$ instances of Alg.~\ref{a:componentreduction} 
simultaneously.
As Alg.~\ref{a:componentreduction} is executed on a graph
with degree $O(\log\log n)$, thus the number of edges 
in each instance is $O(n\log\log n)$.

As there are $\sqrt{n}$ bosses and one coordinator 
in the ``original'' Alg.~\ref{a:componentreduction},
the number of bosses increases to $\sqrt{n}\cdot\sqrt{n}=O(n)$
and the number of coordinators to $\sqrt{n}$
when $\sqrt{n}$ instances are executed simultaneously.
In \kn{step~\ref{ss:send:edge} of Alg.~\ref{a:componentreduction}, 
the node $u_j$ is supposed to deliver each of $|N(u_j)|$ incident edges 
with probability $1/\log\log n$ to each of the bosses
in an instance
of the \CC{} problem. }
Hence, $u_j$ has at most
$\sqrt{n}|N(u_j)|=O(n)$ messages to send in all $\sqrt{n}$ instances
of the \CC{} problem.
As each edge is sent to each boss independently
with probability $1/\log\log n$ and all graphs have $O(n\log\log n)$
edges,
each boss receives 
$$O\left(\frac1{\log\log n}\cdot n\log\log n\right)=O(n)$$
edges with high probability, by a standard Chernoff bound.
\kn{Thus, we can perform step~\ref{ss:send:edge} in parallel for
all instances using Lenzen routing lemma \cite{Lenzen:2013:ODR:2484239.2501983}.}
In the remaining steps of Alg.~\ref{a:componentreduction}:
\begin{itemize}
\item
the original nodes $u_j$ send a single message to each boss or to
the coordinator;
\item
the bosses send a message to all nodes of the network.
\end{itemize}
As there are $O(n)$ bosses and coordinators (each of them
participating in exactly one instance of the \CC{} problem), all $\sqrt{n}$
executions can be performed without asymptotic slowdown
of the algorithm.

\subsection{Parallel executions of Algorithm~\ref{a:cc}}\labell{ss:parallel:final}

Using Lemma~\ref{l:parallel1}, we can execute step~\ref{s:cc1} of Alg.~\ref{a:cc}
in parallel for all instances from Lemma~\ref{l:red2:hegeman}.
However, after such execution, the results are distributed among proxies, not
available in the original nodes of the network corresponding to the nodes of the graph.
It might seem that, in order to perform the remaining steps of $m=\sqrt{n}$ 
instances of Alg.~\ref{a:cc}, we can collect
appropriate information back at the ``original'' nodes and use 
Lemma~\ref{l:GP2016} and Lemma~\ref{l:parallel2}.
This however is not that simple, as we have to face the following obstacle.
The original instances $G_i(V,E_i)$ satisfied the relationship $E_1\subseteq\cdots\subseteq E_m$,
which helped to pass information about neighborhoods to the proxies.
After executions of Alg.~\ref{a:degreereduction}, this ``inclusion property''
does not hold any more.
Therefore, we continue with proxies: both executions of GPReduction
as well as an execution of Alg.~\ref{a:componentreduction} are executed
in parallel in such a way that proxies work on behalf of their
``master'' nodes. 
Thanks to the fact that each node is a proxy of $c=O(1)$ ``original'' nodes only,
Lemma~\ref{l:GP2016} and Lemma~\ref{l:parallel2} can be applied here
and give $O(1)$ round solutions for all instances with help of Lemma~\ref{l:auxiliary}.
\comment{
Lemma~\ref{l:GP2016} and Lemma~\ref{l:parallel2} can be applied here as follows.
We can simulate each round of the original executions of GPReduction
and Alg.~\ref{a:componentreduction} in $c^2$ rounds
indexed by pairs $(a,b)\in[c]^2$. In the round $(a,b)$ each proxy acts 
as a transmitter (receiver, resp.) on behalf
of the $a$th ($b$th resp.) node assigned to it. (See Lemma~\ref{l:auxiliary} for more
detailed exposition of this technique.)
Then, each node in each round is supposed to perform a subset of transmissions
of the node ``simulated'' by it in the execution performed by the original
nodes and therefore the parallel simulations work by 
Lemma~\ref{l:GP2016} and Lemma~\ref{l:parallel2}.
}
%
%
%
%
Finally, when spanning forests of all $m=\sqrt{n}$ instances of the \CC{} problem
are determined, each proxy knows the parent of the ``simulated'' node in the respective spanning trees.
As we consider $m=\sqrt{n}$ simulations and each node is the proxy of $O(1)$
nodes of the input network, information about parents of nodes in the spanning
trees can be delivered from the proxies to the original nodes in $O(1)$ rounds using 
Lenzen's routing algorithm \cite{Lenzen:2013:ODR:2484239.2501983} (Lemma~\ref{l:lenzen}).

\section{Conclusions}

In the paper, we have established $O(1)$ round complexity
for randomized algorithms solving \MST{} in the congested clique. 
In contrast to recent progress on randomized algorithms for
MST, the best deterministic solution has not been
improved from 2003 \cite{Lotker:2003:MCO:777412.777428}.
As shown in \cite{DBLP:conf/fsttcs/PS16}, the $O(\log^* n)$ round
MST algorithm from \cite{GhaffariParter2016} can be implemented
with relatively small number of messages transmitted over an
execution of an algorithm.
We believe that our solution might also be optimized in 
a similar way. However, to obtain such a result, more
refined analysis and adjustment of parameters are
necessary. 

\kn{An interesting research direction is also to study limited
variants of the general congested clique model, better
adjusted to real architectures.}

\section*{Acknowledgments}
The work of the first author was supported by the Polish National Science Centre grant
			DEC-2012/07/B/ST6/01534.

\ifarxiv

\else
\bibliographystyle{abbrv}
\bibliography{references}
\fi

\end{document}